\documentclass{article}
\usepackage{etex}
\usepackage{amsmath}
\usepackage{tikz}
\usepackage{tkz-graph}
\usepackage{bussproofs}
\usepackage{graphicx}

\usetikzlibrary{arrows,automata}
\usetikzlibrary{arrows,shapes,positioning}
\usetikzlibrary{tikzmark,decorations.pathreplacing,calc}
\newtheorem{theorem}{Theorem}

\newtheorem{claim}[theorem]{Claim}

\newtheorem{corollary}[theorem]{Corollary}

\newtheorem{definition}[theorem]{Definition}

\newtheorem{lemma}[theorem]{Lemma}

\newtheorem{partial solution}[theorem]{Partial Solution}

\newenvironment{proof}[1][Proof]{\textbf{#1.} }{\ \rule{0.5em}{0.5em}}
\newcommand{\imply}{\rightarrow}
\input{tcilatex}

\begin{document}

L. Gordeev, E. H. Haeusler\smallskip

\begin{center}
{\Large Proof Compression and NP Versus PSPACE II\smallskip :}

{\Large Addendum }%
\marginpar{
October 2020}\medskip
\end{center}

\textbf{Abstract. }{\small In \cite{GH2} we proved the conjecture \textbf{NP
= PSPACE} by advanced proof theoretic methods that combined Hudelmaier's
cut-free sequent calculus for minimal logic (HSC) \cite{Hud}\ with the
horizontal compressing in the corresponding minimal Prawitz-style natural
deduction (ND) \cite{Prawitz}. In this Addendum we show how to prove a
weaker result \textbf{NP = coNP} without referring to HSC. The underlying
idea (due to the second author) is to omit full minimal logic and compress
only ``naive'' normal tree-like ND refutations of the existence of
Hamiltonian cycles in given non-Hamiltonian graphs, since the Hamiltonian
graph problem in NP-complete. Thus, loosely speaking, the proof of \textbf{%
NP = coNP} can be obtained by HSC-elimination from our proof of \textbf{NP =
PSPACE} \cite{GH2}.}

\section{Introduction}

Recall that in \cite{GH1}, \cite{GH2} we proved that intuitionistically
valid purely implicational formulas $\rho $ have dag-like ND proofs $%
\partial $ whose weights (= the total numbers of symbols) are polynomial in
the weights $\left| \rho \right| $ of $\rho $. $\partial $ were defined by a
suitable two-folded horizontal compression of the appropriate tree-like ND $%
\partial _{1}$ obtained by standard conversion of basic tree-like HSC proofs 
$\pi $ existing by the validity of $\rho $. We observed that the height and
the total weight of distinct formulas occurring in ($\pi $, and hence also) $%
\partial _{1}$ are both polynomial in $\left| \rho \right| $. From this we
inferred that the compressed dag-like ND proofs $\partial $ are
weight-polynomial in $\left| \rho \right| $. Moreover, it is readily seen
that the latter conclusion holds true for any tree-like ND $\partial
^{\prime }$ with the polynomial upper bounds on the height and total weight
of distinct formulas used. We just arrived at the following conclusion,
where \textsc{NM}$_{\rightarrow }$ is standard purely implicational ND for
minimal logic (see also Appendix).

\begin{definition}
Tree-like\emph{\ }\textsc{NM}$_{\rightarrow }$-deduction with the
root-formula (= conclusion) $\rho $ is \emph{polynomial}, resp. \emph{%
quasi-polynomial}, if\ its weight (= total number of symbols), resp. height
plus total weight of distinct formulas,\ is polynomial in the weight of
conclusion, $\left| \rho \right| $.
\end{definition}

\begin{definition}
A given (tree- or dag-like) \textsc{NM}$_{\rightarrow }$-deduction is called
a \emph{proof} of its root-formula $\rho $ iff every maximal thread
connecting $\rho $ with a leaf $\alpha $ is closed (= discharged), i.e. it
contains a $\left( \rightarrow I\right) $ with conclusion $\alpha
\rightarrow \beta $, for some $\beta $.
\end{definition}

\begin{theorem}
Any quasi-polynomial tree-like proof of $\rho $ is compressible into a
polynomial dag-like proof of $\rho $.
\end{theorem}

Now let $P$ be a chosen NP-complete problem and suppose that $\rho $ is
valid iff $P$ has no positive solution. Then the existence of a tree-like ND
proof $\partial ^{\prime }$ as above will infer the existence of a dag-like
ND proof $\partial $ whose weight is polynomial in $\left| \rho \right| $,
which will eventually imply \textbf{NP = coNP}. In particular, let $P$ be
the Hamiltonian graph problem and purely implicational formula $\rho $
express in standard form that a given graph $G$ has no Hamiltonian cycles.
Suppose that the canonical proof search of $\rho $ in \textsc{NM}$%
_{\rightarrow }$ yields a normal tree-like proof $\partial ^{\prime }$ whose
height is polynomial in $\left| G\right| $ (and hence $\left| \rho \right| $%
), provided that $G$ is non-Hamiltonian. Since normal ND proofs satisfy the
subformula property, such $\partial ^{\prime }$ will obey the requested
polynomial upper bounds in question, and hence the weight of its horizontal
dag-like compression $\partial $ will be polynomially bounded, as desired.
That is, we argue as follows.

\begin{lemma}
Let $P$ be the Hamiltonian graph problem and purely implicational formula $%
\rho $ express in standard form that a given graph $G$ has no Hamiltonian
cycles. There exists a normal (and hence quasi-polynomial) tree-like proof
of $\rho $ whose height is polynomial in $\left| G\right| $ (and hence $%
\left| \rho \right| $), provided that $G$ is non-Hamiltonian.
\end{lemma}

Recall that polynomial ND proofs (whether tree- or dag-like) have
time-polynomial certificates (\cite{GH2}: Appendix), while the
non-hamiltoniacy of simple and directed graphs is coNP-complete. Hence
Theorem 3 yields

\begin{corollary}
$\mathbf{NP=coNP}$\textbf{\ }holds true.
\end{corollary}

This argument does not refer to sequent calculus. Summing up, in order to
complete our HSC-free proof of \textbf{NP = coNP}\ it will suffice to prove
Lemma 4. This will be elaborated in the rest of the paper.

\section{Hamiltonian problem}

Consider a simple\footnote{{\footnotesize A simple graph has no multiple
edges. For every pair of nodes }$(v_{1},v_{2})${\footnotesize \ in the graph
there is at most one edge from }$v_{1}${\footnotesize \ to }$v_{2}$.}
directed graph $G=\langle V_{G},E_{G}\rangle $, $card\left( V_{G}\right) =n$%
. A \emph{Hamiltonian path} (or \emph{cycle}) in $G$ is a sequence of nodes $%
\mathcal{X}=v_{1}v_{2}\ldots v_{n}$, such that, the mapping $i\mapsto v_{i}$
is a bijection of $\left[ n\right] =\{1,\cdots ,n\}$ onto $V_{G}$ and for
every $0<i<n$ there exists an edge $(v_{i},v_{i+1})\in E_{G}$. The
(decision) problem whether or not there is a Hamiltonian path in $G$ is
known to be NP-complete (cf. e.g. \cite{arora}). If the answer is YES then $%
G $ is called Hamiltonian. In order to verify that a given sequence of nodes 
$\mathcal{X}$, as above, is a Hamiltonian path it will suffice to confirm
that:

\begin{enumerate}
\item  There are no repeated nodes in $\mathcal{X}$,

\item  No element $v\in V_{G}$ is missing in $\mathcal{X}$,

\item  For each pair $\left\langle v_{i}v_{j}\right\rangle $ in $\mathcal{X}$
there is an edge $(v_{i},v_{j})\in E_{G}$.
\end{enumerate}

It is readily seen that the conjunction of $1,2,3$ is verifiable by a
deterministic TM in $n$-polynomial time. Consider a natural formalization of
these conditions (cf. e.g. \cite{arora}) in propositional logic with one
constant $\bot $ (\emph{falsum}) and three connectives $\wedge $, $\vee $, $%
\rightarrow $ (as usual $\lnot F:=F\rightarrow \bot $).

\begin{definition}
For any $G=\langle V_{G},E_{G}\rangle $, $card(V_{G})=n>0$, as above,
consider propositional variables $X_{i,v}$, $i\in \left[ n\right] $, $v\in
V_{G}$. Informally, $X_{i,v}$ should express that vertex $v$ is visited in
the step $i$ in a path on $G$. Define propositional formulas $A-E$ as
follows and let $\alpha _{G}:=A\wedge B\wedge C\wedge D\wedge E$.

\begin{enumerate}
\item  \label{A} $A=\bigwedge_{v\in V}\left( X_{1,v}\vee \ldots \vee
X_{n,v}\right) $ (: every vertex is visited in $X$).

\item  \label{B} $B=\bigwedge_{v\in V}\bigwedge_{i\neq j}\left(
X_{i,v}\rightarrow \left( X_{j,v}\rightarrow \bot \right) \right) $ (: there
are no repetitions in $X$).

\item  \label{C} $C=\bigwedge_{i\in \left[ n\right] }\bigvee_{v\in V}X_{i,v}$
(: at each step at least one vertex is visited).

\item  \label{D} $D=\bigwedge_{v\neq w}\bigwedge_{i\in \left[ n\right]
}\left( X_{i,v}\rightarrow \left( X_{i,w}\rightarrow \bot \right) \right) $
(: at each step at most one vertex is visited).

\item  \label{E} $E=\bigwedge_{(v,w)\not\in E}\bigwedge_{i\in \left[ n-1%
\right] }\left( X_{i,v}\rightarrow \left( X_{i+1,w}\rightarrow \bot \right)
\right) $ (: if there is no edge from $v$ to $w$ then $w$ can't be visited
immediately after $v$).
\end{enumerate}
\end{definition}

Thus $G$ is Hamiltonian iff $\alpha _{G}$ is satisfiable. Denote by $%
SAT_{Cla}$ the set of satisfiable formulas in classical propositional logic
and by $TAUT_{Int}$ the set of tautologies in the intuitionistic one. Then
the following conditions hold: (1) $G$ is non-hamiltonian iff $\alpha
_{G}\not\in SAT_{Cla}$, (2) $G$ is non-Hamiltonian iff $\lnot \alpha _{G}\in
TAUT_{Cla}$, (3) $G$ is non-Hamiltonian iff $\lnot \alpha _{G}\in TAUT_{Int}$%
. Glyvenko's theorem yields the equivalence between (2) to (3). Hence $G$ is
non-Hamiltonian iff there is an intuitionistic proof of $\lnot \alpha _{G}$.
Such proof is called a certificate for the non-hamiltoniacy of $G$. \cite
{Statman79} (also \cite{Haeusler2014}) presented a translation from formulas
in full propositional intuitionistic language into the purely implicational
fragment of minimal logic whose formulas are built up from $\rightarrow $
and propositional variables. This translation employs new propositional
variables $q_{\gamma }$ for logical constants and complex propositional
formulas $\gamma $ (in particular, every $\alpha \vee \beta $ and $\alpha
\wedge \beta $ should be replaced by $q_{\alpha \vee \beta }$ and $q_{\alpha
\wedge \beta }$, respectively) while adding implicational axioms stating
that $q_{\gamma }$ is equivalent to $\gamma $ . For any propositional
formula $\gamma $, let $\gamma ^{\star }$ denote its translation into purely
implicational minimal logic in question. Note that $size\left( \gamma
^{\star }\right) \leq size^{3}\left( \gamma \right) $. Now $\gamma \in
TAUT_{Int}$ iff $\gamma ^{\star }$ is provable in the minimal logic.
Moreover, it follows from \cite{Statman79}, \cite{Haeusler2014} that for any
normal ND proof $\partial $ of $\gamma $ there is a normal proof $\partial
_{\rightarrow }$ of $\gamma ^{\star }$ in the corresponding ND system for
minimal logic, \textsc{NM}$_{\rightarrow }$, such that $height\left(
\partial _{\rightarrow }\right) =\mathcal{O}\left( height\left( \partial
\right) \right) $. Thus in order to prove Lemma 4 it will suffice to
establish

\begin{claim}
$G$ is non-hamiltonian iff there exists a normal intuitionistic tree-like ND
proof of $\alpha _{G}\rightarrow \bot $ , i.e. $\lnot \alpha _{G}$, whose
height is polynomial in $n$.
\end{claim}

\subsection{Proof of Claim 7}

The sufficiency trivially follows from the soundness of ND. Consider the
necessity. In the sequel we suppose that a non-Hamiltonian graph $G$ is
fixed and $\alpha _{G}=A\wedge B\wedge C\wedge D\wedge E$ (cf. Definition
6). Let $p:\{1,\ldots ,n\}\mapsto V_{G}$ be any sequence of nodes from $%
V_{G} $ of the length $n$ and let $\mathcal{X}_{p}:=\left\{
X_{1,p[1]},\cdots ,X_{n,p[n]}\right\} $ be corresponding set of
propositional variables. $\mathcal{X}_{p}$ and $p$ represent a path in $G$
that starts by visiting vertex $p[1]$, encoded by $X_{1,p[1]}$, followed by $%
p[2]$, encoded by $X_{2,p[2]}$, etc., up to $p[n]$ encoded by $X_{n,p[n]}$.
Since $G$ is non-Hamiltonian, $\mathcal{X}_{p}$ is inconsistent with $\alpha
_{G}$.

\begin{lemma}
For any $p$\ and $\mathcal{X}_{p}$ as above there is a normal intuitionistic
tree-like ND $\Pi _{p}$ with conclusion $\bot $, assumptions from $\mathcal{X%
}_{p}\cup \left\{ \alpha _{G}\right\} $ and $\emph{height}\left( \Pi
_{p}\right) =\mathcal{O}\left( n^{2}\right) $ : 
\begin{prooftree}
  \AxiomC{$\mathcal{X}_{p}\cup \left\{ \alpha _{G}\right\}$}
  \noLine
  \UnaryInfC{$\Pi_p$}
  \noLine
  \UnaryInfC{$\bot $}
\end{prooftree}
\end{lemma}

\begin{proof}
$\Pi _{p}$ is defined as follows. Since $G$ is non-Hamiltonian, we observe
that at least one of the conditions 1, 3 to be a Hamiltonian path\ (see
above in \S 2) fails for $\mathcal{X}_{p}$. Hence at least one of the
following is the case.

\begin{description}
\item  \emph{There are repeated nodes.} There are $1\leq i<j\leq n$, such
that $p[i]=p[j]=v\in V_{G}$; let $i<j$ be the least such pair. Consider a
deduction $\Gamma _{p}$ : 
\begin{prooftree}
    \AxiomC{$X_{j,v}$}
    \AxiomC{$X_{i,v}$}
    \AxiomC{$X_{i,v}\imply (X_{j,v}\imply \bot)$}
    \BinaryInfC{$X_{j,v}\imply \bot$}
    \BinaryInfC{$ \bot$}
  \end{prooftree}
of $\bot $ from $X_{i,v}$, $X_{j,v}$ and $X_{i,v}\rightarrow
(X_{j,v}\rightarrow \bot )$. Since $\left\{ X_{i,v},X_{j,v}\right\} \subset 
\mathcal{X}_{p}$, the assumption $X_{i,v}\rightarrow (X_{j,v}\rightarrow
\bot )$ is a component of the conjunction $B$ from $\alpha _{G}$. So let $%
\Delta _{p}$ be a chain of $\wedge $-elimination rules deducing $%
X_{i,v}\rightarrow (X_{j,v}\rightarrow \bot )$ from $\alpha _{G}$. Now let $%
\Pi _{p}$ be the corresponding concatenation $\Delta _{p}\circ \Gamma _{p}$
deducing $\bot $ from $\left\{ X_{i,v},X_{j,v},\alpha _{G}\right\} \subset 
\mathcal{X}_{p}\cup \left\{ \alpha _{G}\right\} $. Clearly $\Pi _{p}$ is
normal and $height\left( \Pi _{p}\right) =\mathcal{O}\left( n^{2}\right) $.

\item  \emph{There is a missing edge}. There is $1\leq i<n$, such that $%
p[i]=v\in V_{G}$, $p[i+1]=w\in V_{G}$ and $(v,w)\notin E_{G}$; Let $i$ be
the least such number. Consider a deduction $\Gamma _{p}$ : 
\begin{prooftree}
    \AxiomC{$X_{i+1,w}$}    
    \AxiomC{$X_{i,v}$}
    \AxiomC{$X_{i,v}\imply (X_{i+1,w}\imply \bot)$}
    \BinaryInfC{$X_{i+1,w}\imply \bot$}
    \BinaryInfC{$ \bot$}
  \end{prooftree}
of $\bot $ from $X_{i,v}$, $X_{i+1,w}$ and $X_{i,v}\rightarrow
(X_{i+1,w}\rightarrow \bot )$. Since $\left\{ X_{i,v},X_{i+1,w}\right\}
\subset \mathcal{X}_{p}$ and $(v,w)\notin E_{G}$, the assumption $%
X_{i,v}\rightarrow (X_{i+1,w}\rightarrow \bot )$ is a component of the
conjunction $E$ from $\alpha _{G}$. So let $\Delta _{p}$ be a chain of $%
\wedge $-elimination rules deducing $X_{i,v}\rightarrow
(X_{i+1,w}\rightarrow \bot )$ from $\alpha _{G}$. Now let $\Pi _{p}$ be the
corresponding concatenation $\Delta _{p}\circ \Gamma _{p}$ deducing $\bot $
from $\left\{ X_{i,v},X_{i+1,w},\alpha _{G}\right\} \subset \mathcal{X}%
_{p}\cup \left\{ \alpha _{G}\right\} $. Clearly $\Pi _{p}$ is normal and $%
height\left( \Pi _{p}\right) =\mathcal{O}\left( n^{2}\right) $.
\end{description}
\end{proof}

In the sequel for the sake of brevity we let $V_{G}=\{1,\cdots ,n\}$. Now
consider the deductions $\Pi _{p}^{i}$, $1\leq i\leq n$, in the extended ND
that includes standard $n$-ary $\vee $-elimination rules. $\Pi _{p}^{i}$ are
defined by recursion on $i$ using (in the initial case $i=1$) the $\Pi _{p\left( 1/k\right)}$
from the last lemma, where sequences $p\left( -j\right) :\{1,\ldots
,n\}\mapsto V_{G}\cup \{0\}$ and $p\left( j/k\right) :\{1,\ldots ,n\}\mapsto
V_{G}$\ are defined by $p\left( -j\right) \left[ k\right] :=\left\{ 
\begin{array}{lll}
p\left[ k\right] , & if & k=j, \\ 
0, & else, & 
\end{array}
\right. $ and $p\left( j/k\right) :=\left\{ 
\begin{array}{lll}
p\left[ k\right] , & if & k=j, \\ 
p\left[ j\right] , & else. & 
\end{array}
\right. $ So let

\begin{center}
$\Pi _{p}^{1}=${\small 
\begin{prooftree}
  \AxiomC{$ X_{1,v_1}\!\lor\cdots\lor X_{1,v_n}$}
  \AxiomC{$\mathcal{X}_{p(-1)}\!\cup \!\left\{ \alpha _{G}\right\},[X_{1,v_1}]$}
  \noLine
  \UnaryInfC{$\Pi_{p(1/1)}$}
  \noLine
  \UnaryInfC{$\bot$}
  \AxiomC{$\ldots$}
  \AxiomC{$\mathcal{X}_{p(-1)}\!\cup \!\left\{ \alpha _{G}\right\},[X_{n,v_n}]$}
  \noLine
  \UnaryInfC{$\Pi_{p(1/n)}$}
  \noLine
  \UnaryInfC{$\bot$}
  \QuaternaryInfC{$\bot \ \ \,$}
  \end{prooftree}
}
\end{center}

,

\begin{center}
$\Pi _{p}^{j+1}:=${\small 
\begin{prooftree}
  \AxiomC{$ X_{j+1,v_1}\!\!\lor\!\cdots\!\lor X_{j+1,v_n}$}
  \AxiomC{$\mathcal{X}_{p(-(j+1))}\!\!\cup \!\left\{ \alpha _{G}\right\}\!,\![X_{j+1,v_1}]$}
  \noLine
  \UnaryInfC{$\Pi_{p((j+1)/1)}^{j}$}
  \noLine
  \UnaryInfC{$\bot$}
  \AxiomC{$\!\!\!\!\ldots\!\!\!\!$}
  \AxiomC{$\mathcal{X}_{p(-(j+1))}\!\!\cup \!\left\{ \alpha _{G}\right\}\!,\![X_{j+1,v_n}]$}
  \noLine
  \UnaryInfC{$\Pi_{p((j+1)/n)}^{i}$}
  \noLine
  \UnaryInfC{$\bot$}
  \QuaternaryInfC{$\,\,\,\bot$}
  \end{prooftree}
}
\end{center}

.

Thus for $i=n$ we obtain.

\begin{center}
$\Pi _{p}^{n}=${\small 
\begin{prooftree}
  \AxiomC{$ X_{n,v_1}\!\lor\cdots\lor X_{n,v_n}$}
  \AxiomC{$\mathcal{X}_{p(-(n-1))}\!\!\cup \!\left\{ \alpha _{G}\right\},[X_{n,v_1}]$}
  \noLine
  \UnaryInfC{$\Pi_{p(n/1)}^{n-1}$}
  \noLine
  \UnaryInfC{$\bot$}
  \AxiomC{$\!\!\!\ldots\!\!\!$}
  \AxiomC{$\mathcal{X}_{p(-(n-1))}\!\cup \!\left\{ \alpha _{G}\right\},[X_{n,v_n}]$}
  \noLine
  \UnaryInfC{$\Pi_{p(n/n)}^{n-1}$}
  \noLine
  \UnaryInfC{$\bot$}
  \QuaternaryInfC{$\bot \ \ \,$}
  \end{prooftree}
}
\end{center}

.

\begin{lemma}
For any $p:\{1,\ldots ,n\}\mapsto V_{G}$, $\Pi _{p}^{n}$ is a normal
intuitionistic tree-like deduction with conclusion $\bot $ and (the only)
open assumption $\alpha _{G}$ in the extended ND in question Moreover, $%
\emph{height}\left( \Pi _{p}^{n}\right) =\mathcal{O}\left( n^{2}\right) $.
\end{lemma}

\begin{proof}
This easily follows from Lemma 8 by induction on $n$.
\end{proof}

Now let $\Pi :=\Pi _{Id}^{n}$ where $Id:\{1,\ldots ,n\}\mapsto V_{G}$ is the
identity $Id\left[ i\right] :=i.$ Denote by $\widehat{\Pi }$ the canonical
tree-like embedding of $\Pi $ into basic intuitionistic ND with plain
(binary) $\vee $-eliminations that is obtained by successive unfolding of
the $n$-ary $\vee $-elimination rules with premises $X_{j,v_{1}}\vee \cdots
\vee X_{j,v_{n}}$ involved. Note that $\emph{height}\left( \widehat{\Pi }%
\right) =\mathcal{O}\left( n^{3}\right) $. Moreover let $\partial $ denote $%
\widehat{\Pi }$ followed by the introduction of $\alpha _{G}\rightarrow \bot 
$ :%
\begin{prooftree}
   \AxiomC{$\alpha_{G}$}
  \AxiomC{$[ \alpha _{G} ]$}
  \noLine
   \UnaryInfC{$\widehat{\Pi }$}
    \noLine
    \UnaryInfC{$\bot$}
    \BinaryInfC{$\alpha_{G}\rightarrow \bot $}
  \end{prooftree}

\begin{corollary}
$\partial $ is a normal intuitionistic tree-like ND proof of $\alpha
_{G}\rightarrow \bot $ whose height is polynomial in $n$, as required.
\end{corollary}

\section{Appendix: More on Theorem 3}

\subparagraph{Theorem 3 (cf. Introduction)}

\emph{In standard ND for purely implicational minimal logic, NM}$%
_{\rightarrow }$\emph{, any quasi-polynomial tree-like proof }$\partial $%
\emph{\ of }$\rho $\emph{\ is compressible into a polynomial dag-like proof }%
$\partial ^{\ast }$\emph{\ of }$\rho $\emph{.}

\subparagraph{Proof sketch \protect\footnote{{\footnotesize See \cite{GH2}
for a precise presentation.}}}

The mapping $\partial \hookrightarrow \partial ^{\ast }$ is obtained by a
two-folded horizontal compression\ $\partial \hookrightarrow \partial
^{\flat }\hookrightarrow \partial ^{\ast }$, where $\partial ^{\flat }$ is a
polynomial dag-like deduction in NM$_{\rightarrow }^{\flat }$ that extends NM%
$_{\rightarrow }$ by a new \emph{separation} rule $\left( S\right) $%
\begin{equation*}
\fbox{$\left( S\right) :\dfrac{\overset{n\ times}{\overbrace{\alpha \quad
\cdots \quad \alpha }}}{\alpha \ }\ $($n$ arbitrary) }
\end{equation*}
whose identical premises are understood disjunctively: ``\emph{if at least
one premise is proved then so is the conclusion}'' (in contrast to ordinary
inferences: ``\emph{if all premises are proved then so are the conclusions}%
''). The notion of provability in NM$_{\rightarrow }^{\flat }$ is modified
accordingly such that proofs are locally correct deductions assigned with
appropriate sets of threads that are closed and satisfy special conditions
of \emph{local coherency}. Now $\partial ^{\flat }$ arises from $\partial $
by ascending (starting from the root) merging of different occurrences of
identical formulas occurring on the same level, followed by inserting
instances of $\left( S\right) $ instead of resulting multipremise
inferences. Corresponding locally coherent threads in $\partial ^{\flat }$
are inherited by the underlying (closed) threads in $\partial $ (in contrast
to ordinary local correctness, the local coherency is not verifiable in
polynomial time, as the total number of threads in question might be
exponential in $\left| \rho \right| $). A desired ``cleansed'' NM$%
_{\rightarrow }$-subdeduction $\partial ^{\ast }\subset \partial ^{\flat }$
arises by collapsing $\left( S\right) $ to plain repetitions 
\begin{equation*}
\fbox{$\left( R\right) :\dfrac{\alpha }{\alpha \ }$}
\end{equation*}
with respect to the appropriately chosen premises of $\left( S\right) $. The
choice is made non-deterministically using the set of locally coherent
threads in $\partial ^{\flat }$.

--------------------------------------------------------------------------------------------

--------------------------------------------------------------------------------------------

\textbf{Lew Gordeev}

University of T\"{u}bingen

Department of Computer Science

Sand 14, 72076 T\"{u}bingen,

Nedlitzer Str. 4a, 14612 Falkensee

Germany

\emph{E-mail}:\texttt{\ }l\texttt{ew.gordeew@uni-tuebingen.de\medskip }

\textbf{Edward Hermann Haeusler}

Pontificia Universidade Cat\'{o}lica do Rio de Janeiro - RJ

Department of Informatics

Rua Marques de S\~{a}o Vicente, 224, G\'{a}vea, Rio de Janeiro

Brasil

\emph{E-mail}:\texttt{\ hermann@inf.puc-rio.br\medskip }

\end{document}